\documentclass[journal]{IEEEtran}
\pdfoutput=1
\usepackage{cite}

\usepackage{setspace}
\usepackage{comment}
\usepackage{url}

\usepackage{graphicx}
\usepackage{amssymb}
\usepackage{amsmath,amsfonts,amssymb,euscript,epsfig,psfrag,
enumerate,float,afterpage, subfigure}%

\newtheorem{thm}{Theorem}

\newtheorem{defn}{Definition}
\newtheorem{prop}{Property}
\newtheorem{lem}{Lemma}

\newcommand{\defequiv}{\mbox{\raisebox{-.3ex}{$\overset{\vartriangle}{=}$}}}

\usepackage{cite}

\begin{document}

\title{Channel Assignment in Dense MC-MR Wireless Networks: Scaling Laws and Algorithms} 

\author{Rahul Urgaonkar, Ram Ramanathan, Jason Redi, William N. Tetteh \\
Network Research, Raytheon BBN Technologies, Cambridge, MA 02138 \\
Email: \{rahul, ramanath, redi, wtetteh\}@bbn.com}

\maketitle
\begin{abstract}
 We investigate optimal channel assignment  algorithms that
maximize per node throughput in dense multi-channel multi-radio
(MC-MR) wireless networks. Specifically, we consider
an MC-MR network where all nodes are within the
transmission range of each other.
This situation is encountered in many real-life settings such as 
students in a lecture hall, delegates attending a conference, or soldiers in a battlefield.
In this scenario, we show that intelligent assignment of the available 
channels results in a significantly higher per node throughput. 
We first propose a class of channel assignment algorithms, parameterized by $T$ (the number of transceivers per node), 
that can achieve $\Theta\big(\frac{1}{N^{1/T}}\big)$ per node throughput using $\Theta(TN^{1-\frac{1}{T}})$ 
channels.  In view of practical constraints on $T$, we then propose
another algorithm that can achieve $\Theta\big(\frac{1}{(\log_2 N)^2}\big)$ 
 per node throughput using only two transceivers per node.
Finally,  we identify a fundamental relationship between 
the achievable per node throughput,
the total number of channels used, and
the network size  under any strategy. 
Using analysis and simulations, we show that our algorithms 
achieve close to optimal
performance at different operating points on this curve. 
Our work has several interesting implications on the optimal 
network design for dense MC-MR wireless
networks.
\end{abstract}

\section{Introduction}
\label{section:intro}

Starting with the seminal work in \cite{GK}, there has been a lot of research on studying 
optimal scaling laws for wireless ad hoc networks under different assumptions about node capabilities \cite{aeron, ozgur}, 
availability of infrastructure support \cite{infra, ulas}, mobility \cite{mobility, neely}, channel models \cite{gaussian}, traffic models \cite{luna}, etc.
Given the vast literature in this area, we
refer the interested reader to the excellent surveys \cite{NOW1, NOW2} for an overview of this work.

One practical way to improve the capacity of wireless ad hoc networks is to deploy  MC-MR networks. 
The resulting capacity gains 
are mainly due to the use of additional channels as compared to the single channel model of \cite{GK}.
This has attracted much attention recently and several works study the problem of channel assignment, routing, and scheduling 
in general MC-MR networks under different constraints \cite{kodialam}, \cite{jsac}, \cite{xlin}. 
The work in \cite{kyasanur} extends the analytical framework of
 \cite{GK} and derives scaling laws for general MC-MR networks. It is shown in  \cite{kyasanur} that, depending on the number of
 interfaces per node and the number of available channels, there can be a degradation in the network capacity.

Much of this work on scaling laws has focused on general ad hoc networks where it is assumed that the nodes are randomly distributed in an area and that
multi-hop routing is necessary for end-to-end communication. 
However, one scenario that has received little attention is where all nodes are within the transmission range of each other. 
This scenario, which we refer to as the \emph{Parking Lot} model, 
is encountered in many real-life settings. 
Examples include students in a lecture hall, delegates attending a conference, or a platoon of soldiers in a battlefield.
Because of its prevalence, it is important to study this model in more detail and understand fundamental limits on its performance.

\begin{figure}
\centering
\includegraphics[height=6.9cm, width=6.5cm, angle=0]{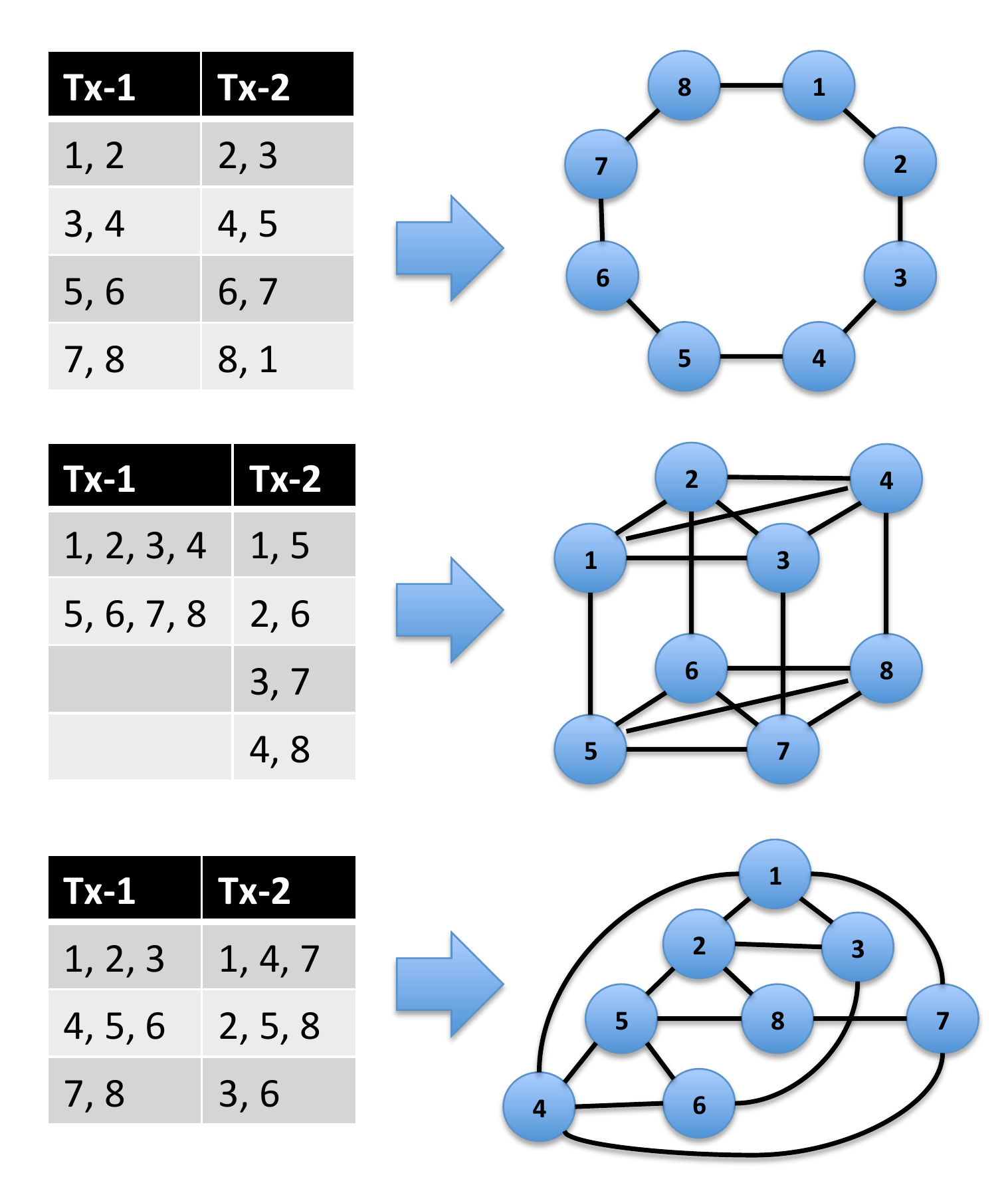}
\caption{Three possible channel assignments in a Parking Lot network of $8$ nodes with $2$ transceivers per node. 
For each transceiver (Tx-1 and Tx-2), the entries in the tables show the list of nodes that share an orthogonal 
channel on that transceiver. The resulting effective topology is shown on the right.}
\label{fig:example_assignment}
\end{figure}

The Parking Lot model is not very interesting when considering single-channel wireless networks because
in this case there is only one feasible solution. 
To ensure connectivity, all transceivers must use the same channel. This results in every node being one hop away from every other node. 
However, when nodes are equipped with multiple transceivers and the network has multiple available channels, 
the problem of channel assignment to different transceivers to maximize per node throughput becomes non-trivial. 
In this case, to increase network capacity, the transceivers must be assigned to different channels. 
This can result in a network graph that effectively has a multi-hop topology even though
all nodes are within each other's transmission range and could reach each other in one hop if they were using the same channel.
Further, by choosing different assignments, a variety of such effective network topologies can be 
formed. For example, Fig. \ref{fig:example_assignment} shows
three possible channel assignments to the transceivers in an $8$ node 
Parking Lot network with $2$ transceivers per node. Also shown is the resulting topology for each assignment. 
Note that these are only three out of many possible channel assignments for this network.

This example illustrates that the presence of multiple transceivers adds a ``degree of freedom'' and expands the network 
design space considerably. Now the problem of throughput maximization involves a \emph{joint} optimization of the channel assignment, 
routing, and scheduling strategies. This is subject to the following tradeoff: 
Given a fixed number of transceivers per node, when a channel assignment scheme uses more channels,
it allows for more concurrent transmissions.
However, the path between any two nodes tends to become 
longer, leading to more relay transmissions and reduced per node throughput. 
To keep the path lengths small, more nodes should share a channel. 
However, this reduces throughput as at most one node can transmit per channel at any time. 
This is because spatial reuse of frequency channels is not possible in the Parking Lot model. 

Thus, the key design objective is to maximize the number of concurrent transmissions while keeping the path lengths small.
However, this by itself is not sufficient. 
For example, one could assign the first transceiver of all nodes to the same channel and
form many small groups using the remaining transceivers. 
With this assignment, all nodes are within one hop of each other while many concurrent transmissions are possible. 
However, the common channel of the first transceiver becomes a bottleneck. What is needed is to ensure that there is sufficient capacity for
\emph{all} flows in the network.

In this paper, we study the problem of optimal channel assignment, routing, and scheduling in an MC-MR Parking Lot network
where the goal is to maximize the uniform per node throughput. 
In addition, we are interested in characterizing the \emph{scaling} behavior of per node throughput as a function of network size.  
We focus on \emph{static} channel assignment algorithms where a particular assignment, once assigned,
is used for a reasonably long period of time. This is motivated by the practical limitations and resulting overheads associated 
with dynamic channel switching. In this setting, there is rich design space involving the number of available channels,
number of transceivers per node and the network size that we explore.
Our main contributions are:

\begin{itemize}
\item We propose a class of algorithms parameterized by $T$, the number of transceivers per node, 
that can achieve $\Theta\big(\frac{1}{N^{1/T}}\big)$ per node throughput using $\Theta(TN^{1-\frac{1}{T}})$ 
channels in the Parking Lot model. Thus, an algorithm in this class can achieve $\Theta(1)$ per node throughput if $T = \Theta(\log_2 N)$.
 
 \item In view of practical constraints on $T$, we propose
another algorithm that can achieve $\Theta\big(\frac{1}{(\log_2 N)^2}\big)$ 
 per node throughput using only two transceivers per node.
 
 \item We identify a fundamental relationship between the
network size, the total number of channels used, and the
achievable per node throughput under any strategy.  
Our algorithms 
are shown to achieve close to optimal
performance at different operating points on this curve.
\end{itemize}

The rest of the paper is organized as follows. In Sec. \ref{section:network_model}, we present the network model and assumptions along with a discussion of
our analytical approach. Sec. \ref{section:hint} and \ref{section:towards} analyze two channel assignment schemes that can achieve progressively higher
per node throughput by using more channels. 
In Sec. \ref{section:tradeoff}, we establish a fundamental relation between the number of channels used, network size and
per node throughput under {any} scheme. We also define a common metric by which to evaluate the performance of a scheme. 
This allows us to bound the performance of the schemes studied in  Sec. \ref{section:hint} and \ref{section:towards} against the best possible scheme. 
In Sec. \ref{section:eval}, we compare the performance of our channel assignment schemes along with some other schemes using packet-level simulations.
Finally, we discuss the implications of our results on the optimal network architecture for dense MC-MR networks in Sec. \ref{section:implications}.

\section{Network Model}
\label{section:network_model}

\emph{Channel Model}: We consider a wireless network of $N$ nodes, indexed $1, 2, \ldots, N$, where all nodes are within the transmission range of each other.
The network has a total of $F$ orthogonal frequency channels, each of bandwidth $B$ Hz.
Every node in the network is assumed to have $T \geq 2$ transceivers. {As noted earlier, the $T=1$ case has a trivial solution.}
Each transceiver can be assigned at most one channel at any time. 
We assume that the transceivers operate in the half-duplex mode where they can either transmit or receive (but not both) on their channel at any time. 
Further, we assume a collision model for interference so that at most one transceiver can transmit successfully on a channel at any time. 
The link level transmission rate per channel is assumed to be the same for all channels and is given by $R$ bps.

\emph{Traffic Model}: We consider the $N$ source-destination pair random unicast model where 
every node is the source of one unicast session destined to another node that is chosen uniformly at random.
Each source generates packets of size $D$ bits at a rate given by $\lambda$ packets/sec. 
We assume that the source-destination pairings are not known a priori. 
Note that if the pairings were known,
then one can optimize the channel assignments and routing with respect to them. 
This would yield higher per node throughput. 
However, our focus is on the case where the traffic patterns are not known a priori and/or change rapidly. 
Under these assumptions, this traffic model is equivalent to the model where each node generates packets
at rate $\lambda$ packets/sec and each packet is equally likely to be destined to any other node. We are
interested in designing algorithms that maximize the rate $\lambda$ that can be supported. We call this rate the
uniform per node throughput.

\emph{Channel Assignment Model}: 
We focus on static channel assignment schemes where a particular assignment, once assigned, is used for a reasonably long period of time.
Dynamic channel assignment can be shown to outperform static schemes. 
 Indeed, it can be shown to
 achieve the best possible throughput under our model. However, factors such as switching delay, coordination 
 overheads, and hardware constraints make dynamic channel assignment challenging to implement in practice. 


\emph{Objective}: Given this network model, our objective is to design a policy that maximizes the uniform per node throughput. 
Since the channel assignments under such a policy may effectively create a multi-hop topology, the
overall policy consists of the channel assignment, routing as well as transmission scheduling algorithms.
For any given tuple $(N, T, F)$, this can be formulated as an optimization problem that searches over
all possible channel assignment, routing, and scheduling options and yields the maximum throughput. 
However, given the enormous search space, this approach quickly becomes intractable and
does not yield useful insights. 

Instead,  we take a \emph{scaling analysis} approach where
we focus on characterizing the scaling behavior of the achievable throughput as a function of the network size $N$.
This approach is useful because, while being tractable, it provides insights into the optimal network design and helps derive scaling laws. 
In analyzing the scaling behavior of the throughput, we assume that the number of available channels $F$ can scale with $N$ according to a fixed function. 
For example, $F(N)$ could be $\sqrt{N}$.
This approach significantly simplifies the problem and allows us to develop algorithms
whose achievable throughput can be  computed exactly in closed form. 
We will present these algorithms in Sec. \ref{section:hint} and \ref{section:towards}. 
For each algorithm that we study, 
we will also calculate the exact total number of channels used. 
Finally, in Sec. \ref{section:tradeoff}, we will bound the ``performance gap''
 between these algorithms and the best possible throughput that can be achieved 
by solving the optimization problem for $(N, T, F(N))$.

{For simplicity, in the rest of the paper, we assume normalized and idealized transmission rates so that $1$ packet can be transmitted per
channel use. This assumption is relaxed in the simulations in Sec. \ref{section:eval}.

\section{HINT: Hierarchical Interleaved Channel Assignment}
\label{section:hint}

In this section, we present a class of channel assignment strategies that we call {HINT-T} 
which can achieve $\Theta\big(\frac{1}{N^{1/T}}\big)$ 
per node throughput using 
$F = \Theta(N^{\frac{T-1}{T}})$ 
channels and $T$ transceivers per node. 
The main idea behind this strategy is to form $N^{\frac{T-1}{T}}$ groups, each of size $N^{1/T}$, for every transceiver index. 
Then it  assigns transceivers to these groups in such a way that every node can reach every other node in \emph{at most} $T$ hops. 
 For simplicity of presentation, we assume throughout that $N^{1/T}$ is an integer. 
This scheme can be modified to be applicable when this is not the case using similar ideas. 
For brevity, we do not discuss this extension in the paper. However, in the simulations 
in Sec. \ref{section:eval}, we also implement this scheme for the case when $N^{1/T}$ is not an integer.


\subsection{HINT-2}
\label{section:hint2}

To provide intuition, we first describe the assignment for $T=2$. 
The general case is treated in Sec. \ref{section:hintT}.

\subsubsection{Channel Assignment under HINT-2}
\label{section:hint2_CA}

Let $M = N^{1/2}$. Group the first transceiver of all nodes into $M$ groups, each containing $M$ nodes. 
We call these the Tx-1 groups and assign consecutively numbered nodes to them as follows. 
The first Tx-1 group contains nodes $1, 2, \ldots, M$, the second Tx-1 group contains nodes $M+1, M+2, \ldots, 2M$, and so on. 

Similarly, group the second transceiver of all nodes into $M$ groups, each containing $M$ nodes.
We call these the Tx-2 groups and assign nodes to them as follows. 
The first Tx-2 group contains the first node from each Tx-1 group, i.e., nodes $1, M+1, 2M+1, \ldots, (M-1)M + 1$.
Likewise, the second Tx-2 group contains the second node from each Tx-1 group, and so on.

Each of the groups thus formed  is assigned an orthogonal channel. 
Since there are a total of $2M$ groups, the total number of channels used is $2M$.
Note that the first transceiver of every node is assigned to a Tx-1 group while 
the second transceiver of every node is assigned to a Tx-2 group.
Fig. \ref{fig:hint2_example} illustrates this assignment for a network of $N=16$ nodes using
channels $f_1$ to $f_8$. Note that nodes in Tx-2 groups are obtained by \emph{interleaving} nodes from the
Tx-1 groups. As shown later in Sec. \ref{section:hintT}, the same idea can be applied in a \emph{hierarchical}  fashion
to obtain assignments for the general $T$ case. Hence, we call this family of schemes 
HINT-T: Hierarchical Interleaved Channel Assignment with $T$ transceivers.

\begin{figure}
\centering
\includegraphics[height= 5cm, width = 6.5cm, angle=0]{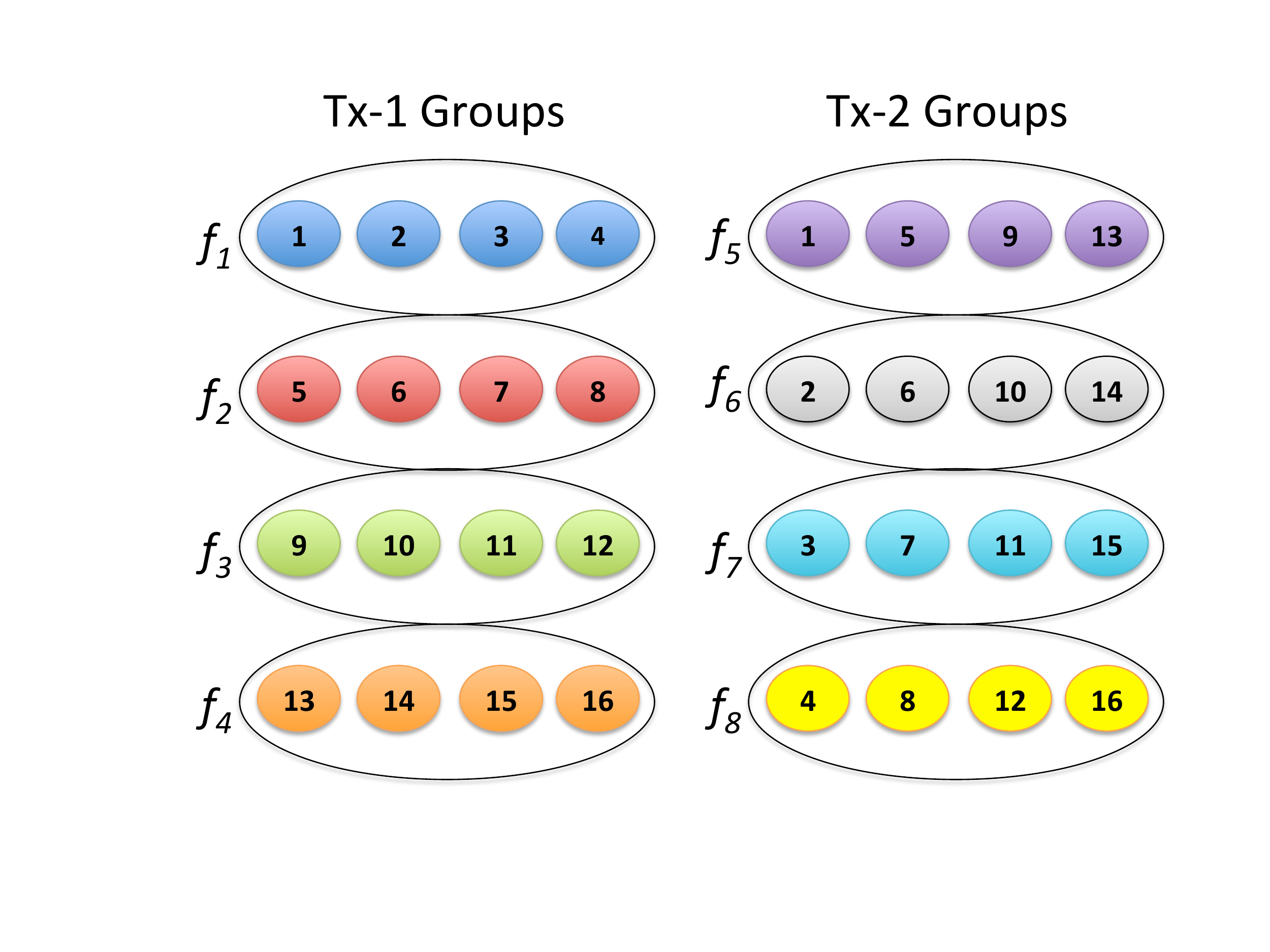}
\caption{Channel Assignments under HINT-2 for  $N =16$.}
\label{fig:hint2_example}
\end{figure}

\subsubsection{Scheduling and Routing under HINT-2}
\label{section:hint2_routing}

We next describe the scheduling and routing strategy that is used with this assignment.

\emph{Scheduling Strategy}: Every transceiver in a group gets the same fraction of time to transmit on that group's channel. 
Since each group has $M$ transceivers, every 
transceiver gets ${1}/{M}$ of the total transmission capacity of the channel.

\emph{Routing Strategy}: For any source node $s$, if a destination node $d$ is in the same Tx-1 group, 
it transmits directly to $d$  in one hop on its Tx-1 channel. 
Else, $s$ transmits to that node $r$ in its Tx-2 group which shares its Tx-1 group with $d$. Node $r$
then forwards to $d$ using its Tx-1 channel.


For example, in Fig. \ref{fig:hint2_example}, the route from node $1$ to node $3$ follows the path $1-3$ on channel $f_1$. 
This involves transmissions only by the first transceiver of node $1$. On the other hand, the route from
node $1$ to node $11$ follows the path $1-9$ on channel $f_5$ and $9-11$ on channel $f_3$. 
This involves transmissions by the second transceiver of node $1$ and the first transceiver of node $9$.
It can be seen that this routing strategy ensures that every node is within $2$ hops of any other node.

\subsubsection{Throughput Analysis of HINT-2}
\label{section:hint2_analysis}

We now show that the HINT-2 assignment along with the scheduling and routing strategy as described above
can achieve a throughput of ${1}/{M}$ for every node. 
\begin{thm} 
For $N^{1/2} = M$ (where $M$ is an integer),
HINT-2 can achieve a uniform per node throughput of $1/N^{1/2}$ using $2N^{1/2}$ channels.
\label{thm:hint2_thruput}
\end{thm}

\begin{proof}
The proof uses the following observation.
Because of the symmetry of the assignment, it is sufficient
to focus on the total load on each of the two transceivers of
node $1$ and show that it can be supported. The details are provided in Appendix A.
\end{proof}

\subsection{HINT-T}
\label{section:hintT}

We next present the assignment strategy for $T > 2$.

\subsubsection{Channel Assignment under HINT-T}
\label{section:hintT_CA}

Let $M = N^{1/T}$. Similar to the $T=2$ case, the HINT-T strategy forms $M^{T-1}$ groups for every transceiver index. 
The groups corresponding to transceiver index $k$ are called Tx-k groups where $1 \leq k \leq T$. 
Each Tx-k group contains $M$ nodes and is assigned one orthogonal channel.
Since there are a total of $TM^{T-1}$ such groups, the total number of channels used under HINT-T is $TM^{T-1}$. 
The node assignment to these groups is performed as follows.

Fix a transceiver index $k$ where $1 \leq k \leq T$. Let $i$ and $j$ be integers such that $1 \leq i \leq M^{T-k}$ and $1 \leq j \leq M^{k-1}$.
Then the Tx-k group number $(i-1)M^{k-1} + j$ contains the following nodes:
\begin{align}
\{ &(i-1)M^k + j, (i-1)M^k  + j + M^{k-1}, \nonumber\\
& (i-1)M^k  + j + 2M^{k-1}, (i-1)M^k + j  + 3M^{k-1},  \nonumber\\
& \ldots, (i-1)M^k  + j + (M-1)M^{k-1} \}.
\label{eq:hintT_1}
\end{align}
This definition ensures that every Tx-k group has $M$ nodes. Further, 
the total number of Tx-k groups is $\sum_{i=1}^{M^{T-k}} \sum_{j=1}^{M^{k-1}} 1 = M^{T-1}$.

Fig. \ref{fig:hint3_example2_table} shows the 
HINT-3 assignment for a network  of $N=27$ nodes with $T=3$ transceivers per node. 
It also illustrates the routing strategy that will be described in Sec. \ref{section:hintT_routing}.
Note that this assignment ensures that every node is within $3$ hops of any other node. 
In general, the HINT-T assignment ensures that every node is within $T$ hops of any other node.

\subsubsection{Level Sets}
\label{section:hintT_level_sets}

In order to describe the routing strategy, 
we define the following collection of sets for each transceiver index $k$ where $1 \leq k \leq T$. 
For $1 \leq i \leq M^{T-k}$, define $\mathcal{S}_{ik}$ as the set containing the nodes $(i-1)M^k + 1$ to $iM^k$. 
For each $k$, there are $M^{T-k}$ such sets, each containing $M^k$ consecutively numbered nodes. 
We call them the ``level $k$ sets''.
Fig. \ref{fig:hint3_example2_table} illustrates these sets for $N=27$.
Note that for  $k=3$, there is only one such set $\mathcal{S}_{13}$
and it consists of all nodes.

The following properties follow directly from the definition of the level sets and the HINT-T assignment in (\ref{eq:hintT_1}).
\begin{enumerate}

\item Any two level sets $\mathcal{S}_{ik}$ and $\mathcal{S}_{jk}$ where $i \neq j$ are disjoint. 

\item Any level set $\mathcal{S}_{ik}$ where $k > 1$ is a union of $M$ level sets from level $k-1$.

\item Under the HINT-T assignment  (\ref{eq:hintT_1}), for $1 \leq k \leq T-1$, no  two nodes from any $\mathcal{S}_{ik}$ are in the same Tx-k+1 group. 
\end{enumerate}

To illustrate the third property in Fig. \ref{fig:hint3_example2_table}, note that nodes $1, 2, 3$ from the first Tx-1 group are in different
Tx-2 groups. Similarly, nodes $3, 6, 9$ from the third  Tx-2 group are in different Tx-3 groups. 
This is precisely the interleaving of nodes from lower Tx groups to form upper Tx groups as seen
in HINT-2. HINT-T generalizes this interleaving to $T$ hierarchical levels.

\begin{figure}
\centering
\includegraphics[width=8.5cm, height=4.5cm, angle=0]{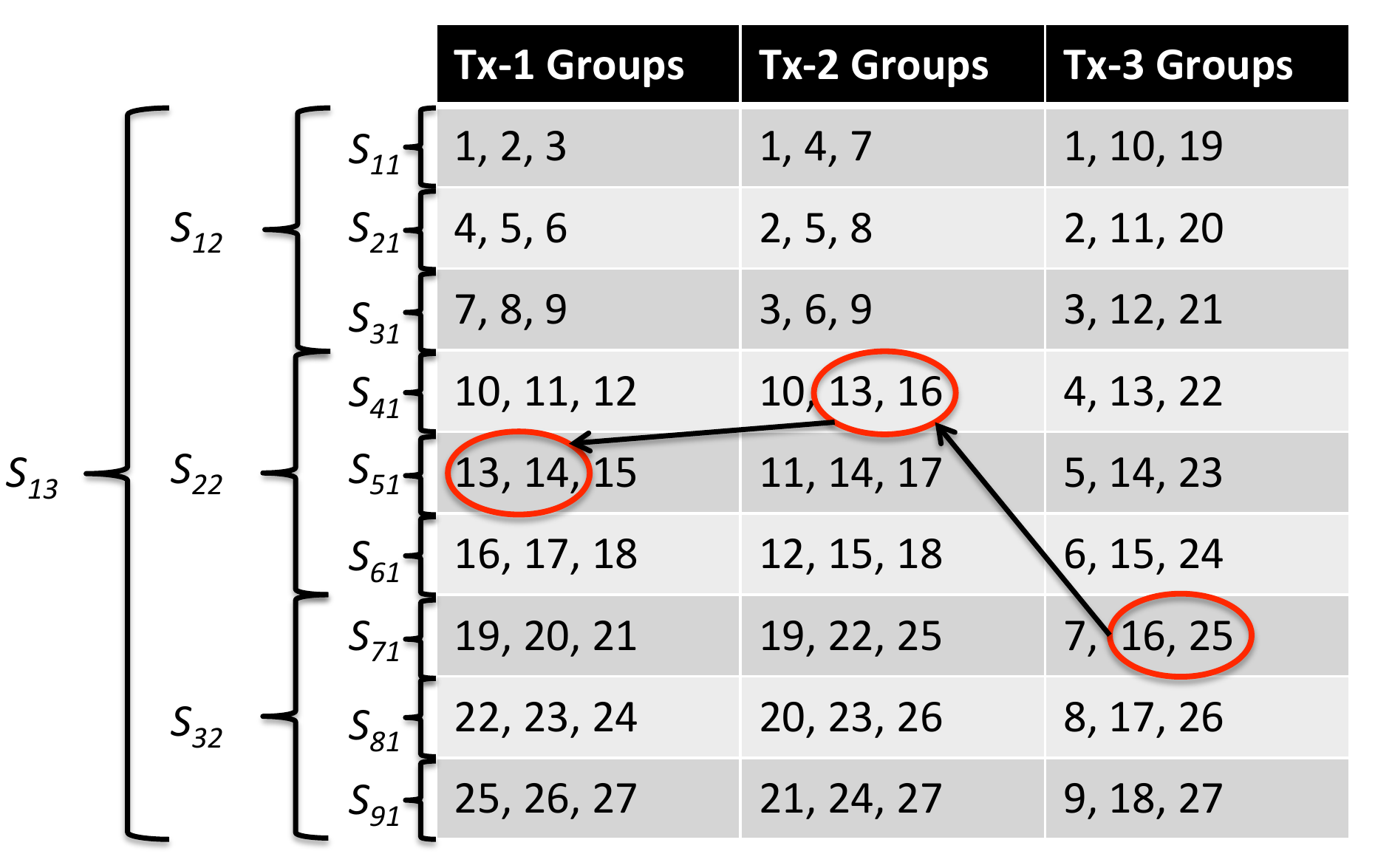}
\caption{Channel Assignment under HINT-3 for $N = 27$. Also shown are the level sets and the route from node $25$ to node $14$.}
\label{fig:hint3_example2_table}
\end{figure}

\subsubsection{Scheduling and Routing under HINT-T}
\label{section:hintT_routing}

We now describe the scheduling and routing strategy that is used with HINT-T.

\emph{Scheduling Strategy}: Similar to HINT-2 strategy, every transceiver in a group gets 
${1}/{M}$ of the total transmission capacity of the channel.

\emph{Routing Strategy}: 
For any two nodes $a$ and $b$, define $k(a, b)$ as the smallest $k$
for which there exists a level set $\mathcal{S}_{ik}$  such that \emph{both} $a$ and $b$ are in $\mathcal{S}_{ik}$. 
Note that at least one such set always exists since the set $\mathcal{S}_{1T}$ contains all nodes.
Further, $k(a, b) \leq T$ for all $a, b$.
The routing strategy from a source node $s$ to a destination node $d$ can be described using these $k(a, b)$ values.

First, calculate $k(s, d)$. If $k(s, d) = 1$, then $d$ is in the same Tx-1 group as $s$ and $s$ 
transmits directly to $d$ in one hop using its first transceiver. 
If $k(s, d) \neq 1$, then $d$ is in a different Tx-1 group than $s$ and the packet is routed as follows.

If $d$ is in the same Tx-k(s, d) group as $s$, then $s$ transmits
directly to $d$ in one hop using its $k(s, d)^{th}$ transceiver.

Else, $s$ determines the node with the smallest value of $k(r, d)$ among all of its neighbors $r$ in its Tx-k(s, d) group. 
Denote this node by $r^*$. 
Then $s$ relays the packet to node $r^*$ using its $k(s, d)^{th}$ transceiver. 
Node $r^*$ now uses the same algorithm as described before to route the packet to $d$.

This procedure is illustrated for the HINT-3 assignment  in Fig. \ref{fig:hint3_example2_table}
where node $25$ wants to send to node $14$. We first calculate $k(25, 14)$. Using Fig.  \ref{fig:hint3_example2_table} , 
it can be seen that $\mathcal{S}_{13}$ is the only set that has both nodes $25$ and $14$. Thus, we have that $k(25, 14) = 3$. 
Next, since $14$ is not a neighbor of  $25$ in its Tx-3 group, node $25$ calculates $k(7, 14)$ and $k(16, 14)$ for its Tx-3 group
neighbors $7$ and $16$. Using Fig.  \ref{fig:hint3_example2_table}, we have $k(7, 14) = 3$ and $k(16, 14) = 2$. Therefore,
node $25$ forwards the packet to node $16$ for relaying to the destination. This procedure is now repeated by node $16$.
The following Lemma characterizes the HINT-T routing strategy.

\begin{lem}
The HINT-T routing strategy ensures that there are at most $T$ hops between any pair of nodes.
\label{lem:hintT_routing}
\end{lem}
\begin{proof} 
 The full proof is  provided in Appendix B.
\end{proof}

\subsubsection{Throughput Analysis of HINT-T}
\label{section:hintT_analysis}

We now show that the HINT-T assignment along with the scheduling and routing strategy as described above
can acheive a throughput of ${1}/{M}$ for every node. 
\begin{thm} 
For $N^{1/T} = M$ (where $M$ is an integer), 
HINT-T achieve a uniform per node throughput of $1/N^{1/T}$ using $TN^{\frac{T-1}{T}}$ channels.
\label{thm:hintT_thruput}
\end{thm}

\begin{proof}
The proof is based on similar ideas as Theorem \ref{thm:hint2_thruput} and is  provided in Appendix C.
\end{proof}

\subsection{Discussion of HINT-T}
\label{section:hintT_discussion}

Theorem \ref{thm:hintT_thruput} implies that if $T= \log_2N$ and if there are $N\log_2N/2$ channels available, then 
HINT-T can achieve a per node throughput given by $1/N^{1/\log_2N} = 1/2$.
Thus, it is possible to get $\Theta(1)$ per node throughput if each node has $\log_2N$ transceivers and
there are $N\log_2N/2$ channels available. This is in sharp contrast to the case of a single radio network, i.e., $T=1$. 
Note that the best possible uniform throughput for $T=1$ under static channel assignment is $1/N$. 
This is because when $T=1$, all transceivers must share the same channel so that the network  remains 
connected. However, when $T > 1$, much higher throughput can be achieved. 
This shows that  there is a fundamental difference between single radio and multi radio networks when static channel
assignment is used.

However, requiring $\Theta(\log_2N)$ transceivers per node is impractical. Also impractical is the
availability of $\Theta(N\log_2N/2)$ channels. This raises the following questions:

\begin{enumerate}

\item Do we really need the number of transceivers per node and total channels required to 
grow to infinity to get $\Theta(1)$ per node throughput? 

\item Is there a fundamental way to characterize the performance of any channel assignment strategy? Specifically, what is the
relationship between number of transceivers, number of channels used, network size and per node throughput?

\end{enumerate}
We address these questions in the following sections.

\section{Towards $\lambda = \Theta(1)$ with $T=2$}
\label{section:towards}

\subsection{LOG-2 Assignment}
\label{section:hint_log}

In this section, we describe a strategy called LOG-2 that achieves $\Theta\big(\frac{1}{(\log_2 N)^2}\big)$ throughput
per node with only two transceivers per node using $O\big(\frac{N}{\log_2N}\big)$ channels.
The main idea behind this strategy is the following. 
We first form two sets of groups, one  per transceiver index.
Each set contains $O\big(\frac{N}{\log_2N}\big)$ groups, each group having size $O(\log_2 N)$ nodes.
Then the strategy  assigns nodes to these groups in such a way that a node can reach any other node in \emph{at most} $O(\log_2 N)$ hops.
For simplicity of presentation,
 we assume in the following that $N$ is of the form $N = M \log_2 M$ where $\log_2M \in \mathbb{Z}^+$.
This scheme can be modified to be applicable when this is not the case using similar ideas. 

\subsubsection{Channel Assignment under LOG-2}
\label{section:hint_log_CA}

The channel assignment is performed as follows. First, group the first transceiver of all nodes 
into $M$ Tx-1 groups, each containing $\log_2M$  consecutively numbered nodes. 
Thus, the $k^{th}$ Tx-1 group contains nodes $(k-1)\log_2M + 1, (k-1)\log_2M + 2, \ldots, k\log_2M$.

Next, group the second transceiver of all nodes into $M$ Tx-2 groups, each containing $\log_2M$ nodes.
We assign nodes to them as follows.
For $1 \leq i \leq \log_2M$, the $i^{th}$ node from Tx-1 group number $((j-1 + 2^{i-1})\bmod M)$ is
assigned to be the $i^{th}$ node of  Tx-2 group number $j$ (where $1 \leq j \leq M$). 
The ``mod \emph{M}'' operation used here and in rest of the paper is defined as follows.
For any non-negative integers $a$ and $b$, we use the following definition:
\begin{displaymath}
(a \bmod M) \defequiv \left\{ \begin{array}{ll}
M & \textrm{if $a = bM$ }\\
x & \textrm{if $a = bM + x$ where $0 < x < M$}
\end{array} \right.
\label{eq:mod_defn}
\end{displaymath}

Each of the groups thus formed  is assigned an orthogonal channel. Since there are a total of $2M$ groups, the total number of channels used is $2M$.
Fig. \ref{fig:hintlog_example_table} shows this assignment for a network of $N=24 = 8 \log_2 8$ nodes with $T=2$. 
To illustrate the working of the algorithm, consider Tx-2 group number $7$. 
For $1 \leq i \leq 3$, the $i^{th}$ node from Tx-1 group number $((7-1 + 2^{i-1})\bmod 8)$ is
assigned to be the $i^{th}$ node of this group. 
For $i=1$, this is given by the first node of Tx-1 group $(6+1 \bmod 8) = 7$, i.e., node $19$.
For $i=2$, this is given by the second node of Tx-1 group $(6+2 \bmod 8) = 8$, i.e., node $23$.
And for $i=3$, this is given by the third node of Tx-1 group $(6+4 \bmod 8) = 2$, i.e., node $6$.

\begin{figure}
\centering
\includegraphics[height=4cm, width=5cm,angle=0]{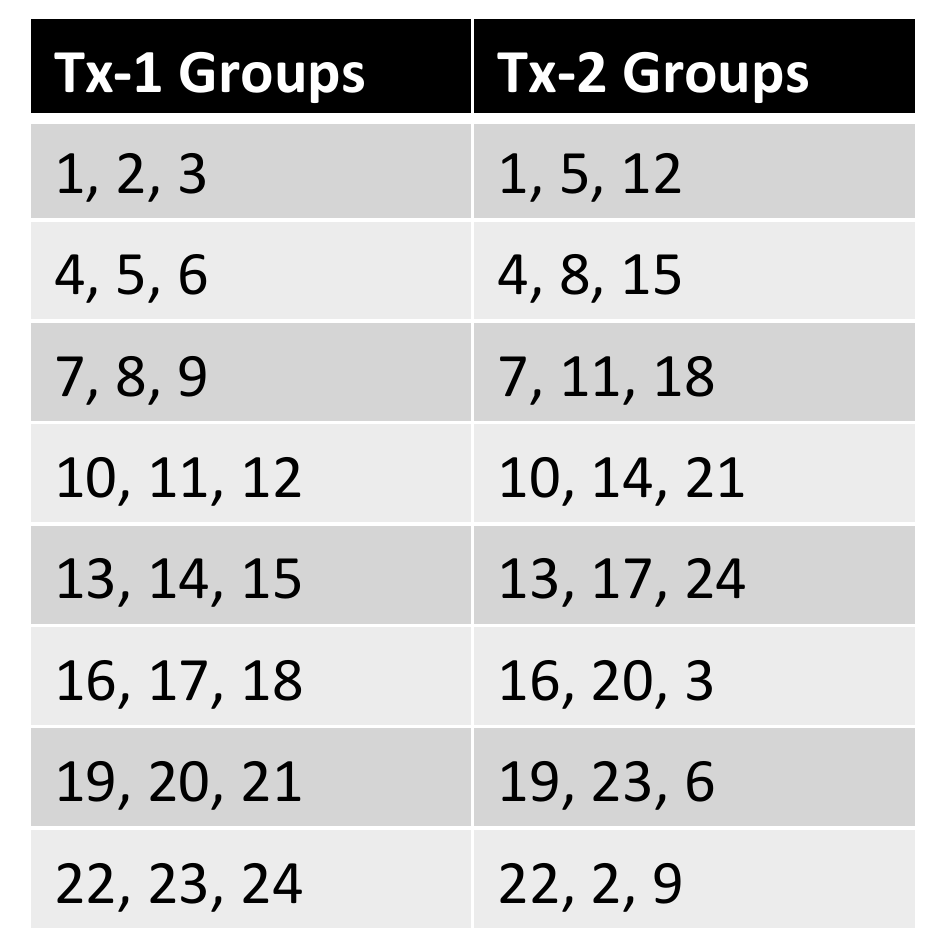}
\caption{Assignment under LOG-2 for a network of size $N = 8 \log_2 8 = 24$.}
\label{fig:hintlog_example_table}
\end{figure}

The following property follows from the definition of the LOG-2 assignment strategy:
\begin{prop}
 In any Tx-2 group number $j$, there is one node from each of the Tx-1 groups numbered 
$((j - 1 + 2^{i-1} )\bmod M)$ where $1 \leq i \leq \log_2M$.
\end{prop}

This means that the difference between the Tx-1 group numbers of consecutive nodes in a Tx-2 group follows the \emph{geometric sequence}
 $2^0, 2^1, 2^2, \ldots, 2^{i-1}$ where $1 \leq i < \log_2M$.  
This is precisely the intuition behind this strategy, as illustrated in Fig. \ref{fig:hintlog_ring}. 
This figure shows the nodes in the first Tx-2 group (black circles) under LOG-2 assignment for $N = 64$ placed on a ring.
It can be seen that the neighboring nodes of $1$ are located progressively farther away as we 
traverse clockwise on the ring. Note that the assignment on Tx-1 ensures that node $1$ can reach close by nodes $2, 3, 4$ in one hop.
As we show in the next section, using a combination of local neighbors on Tx-1 and progressively farther away neighbors on Tx-2, this
strategy ensures that a node is within $2\log_2M + 1$ hops of any other node.
We note that the idea behind this strategy bears a resemblance to prior works on distributed hash tables for fast lookup such as Chord \cite{chord}.

\begin{figure}
\centering
\includegraphics[height=4cm, width=5cm, angle=0]{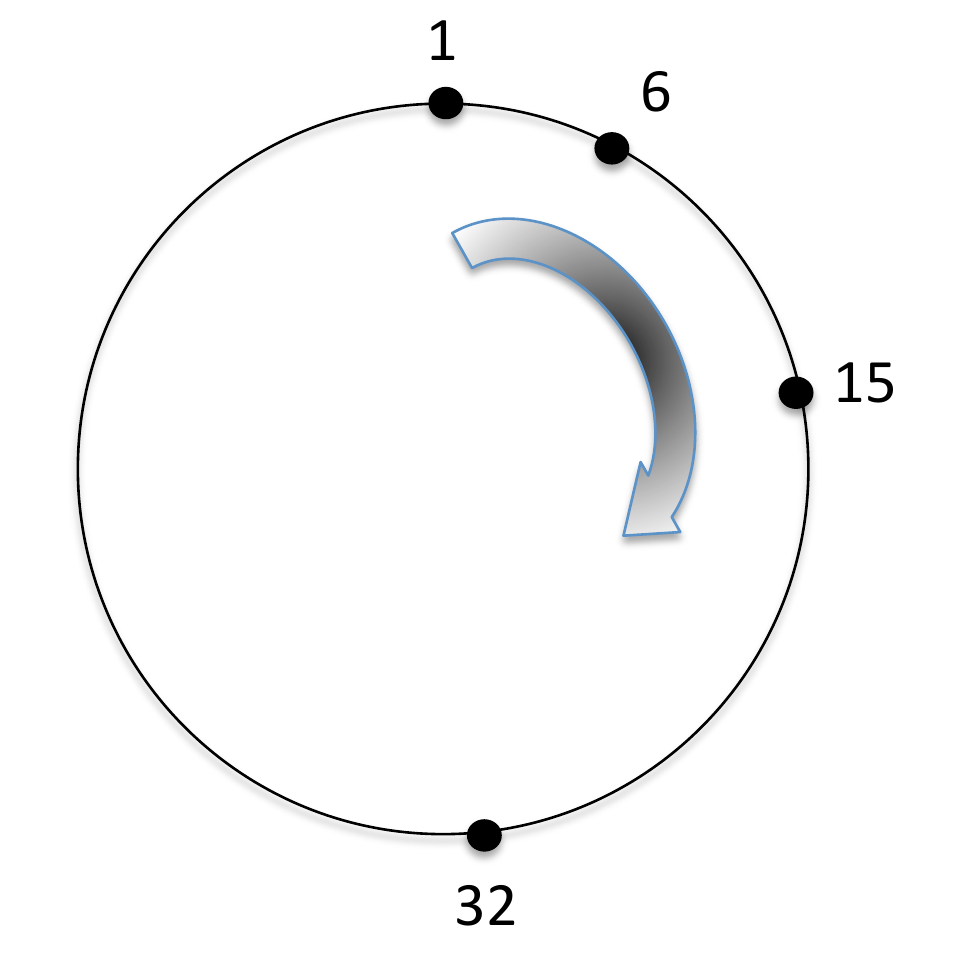}
\caption{Idea behind the formation of Tx-2 groups under LOG-2.}
\label{fig:hintlog_ring}
\end{figure}

\subsubsection{Cover Sets}
\label{section:hint_log_cover_sets}


In order to describe the routing strategy,
we define the following collection of sets for each Tx-2 group $j$ (where $1 \leq j \leq M$).
For every $i^{th}$ node in this group, we define a set $\mathcal{U}_{ij}$ as follows. 
For $1 \leq i \leq \log_2M - 1$, $\mathcal{U}_{ij}$ contains $2^{i-1}$ Tx-1 group numbers, starting from 
$((j-1 + 2^{i-1})\bmod M)$ and incrementing by $1$ and using the ``mod $M$'' operation as defined before. 
For $i=\log_2M$, $\mathcal{U}_{ij}$ contains $(2^{i-1} + 1)$ Tx-1 group numbers, starting from $((j-1 + 2^{i-1})\bmod M)$ 
and incrementing by $1$ while using the mod $M$ operation.
For example, for the network in Fig. \ref{fig:hintlog_example_table}, we have:
\begin{align*}
&\mathcal{U}_{11} = \{1\},  \mathcal{U}_{21} = \{2, 3\}, \mathcal{U}_{31} = \{4, 5, 6, 7, 8\} \\
&\mathcal{U}_{17} = \{7\},  \mathcal{U}_{27} = \{8, 1\}, \mathcal{U}_{37} = \{2, 3, 4, 5, 6\}
\end{align*}
There are a total of $M\log_2M$ such sets.  The routing strategy we discuss next uses these sets as follows. 
 A node at the $i^{th}$ level of Tx-2 group $j$ 
is responsible for relaying to nodes in those Tx-1 groups whose index is in the set $\mathcal{U}_{ij}$. 
Thus, this set lists all those Tx-1 groups that are ``covered'' by this node. Hence, it is called a Cover Set.
For example, in Fig. \ref{fig:hintlog_example_table}, node $5$ which is at the $2^{nd}$ level in  Tx-2 group $1$ 
forwards data for destinations in the set of Tx-1 groups in  $\mathcal{U}_{21}$. Likewise, node $6$ 
which is at the $3^{rd}$ level in Tx-2 group $7$ forwards data for destinations in the Tx-1 groups of  
$\mathcal{U}_{37}$.
Note that for any $j$, the sets  $\mathcal{U}_{ij}$ are disjoint and their union covers all Tx-1 groups.

\subsubsection{Scheduling and Routing under LOG-2}
\label{section:hint_log_routing}
We now describe the scheduling and routing strategy of LOG-2.

\emph{Scheduling Strategy}: Every transceiver in a Tx-1 group gets the same fraction of time to transmit on that group's channel. 
However, only the first node gets all the transmission time in a Tx-2 group. As we will see next, the routing strategy of LOG-2 requires only
the first node of each Tx-2 group to transmit on that group's channel.

\emph{Routing Strategy:} Consider a node $n$ that has a packet destined for node $m$. 
This packet could have been generated by node $n$ itself,
or it could have been forwarded to $n$ to be relayed to $m$. In both cases, node $n$ does the following.
Let the Tx-1 group number of nodes $n$ and $m$ be $g(n)$ and $g(m)$ respectively. 
In order to route a packet from $n$ to $m$, $n$ first checks if $m$ is in its Tx-1 group, i.e., if $g(n) = g(m)$. 
If yes, then $n$ transmits directly to $m$ in one hop using its first transceiver.
Else, $n$ transmits the packet to the first node in Tx-1 group number $g(n)$ for relaying to $n$. Note that 
this step is not required if $n$ itself is the first node in its Tx-1 group.

Let $q$ be the first node in Tx-1 group $g(n)$. Note that under LOG-2,
$q$ will also be the first node in Tx-2 group $g(n)$. 
Node $q$ checks if $m$ is in its Tx-2 group.
If yes, it transmits directly to $m$ in one hop using its second transceiver.
Else, node $q$  transmits the packet to that node in its Tx-2 group that ``covers'' node $m$. 
More precisely, $q$ transmits to the $i^{th}$ node in its Tx-2 group for forwarding to $m$ where
$2 \leq i \leq \log_2M$ and $g(m) \in \mathcal{U}_{ig(n)}$. 
This process is repeated until the packet gets delivered.


\begin{figure}
\centering
\includegraphics[height=4cm,angle=0]{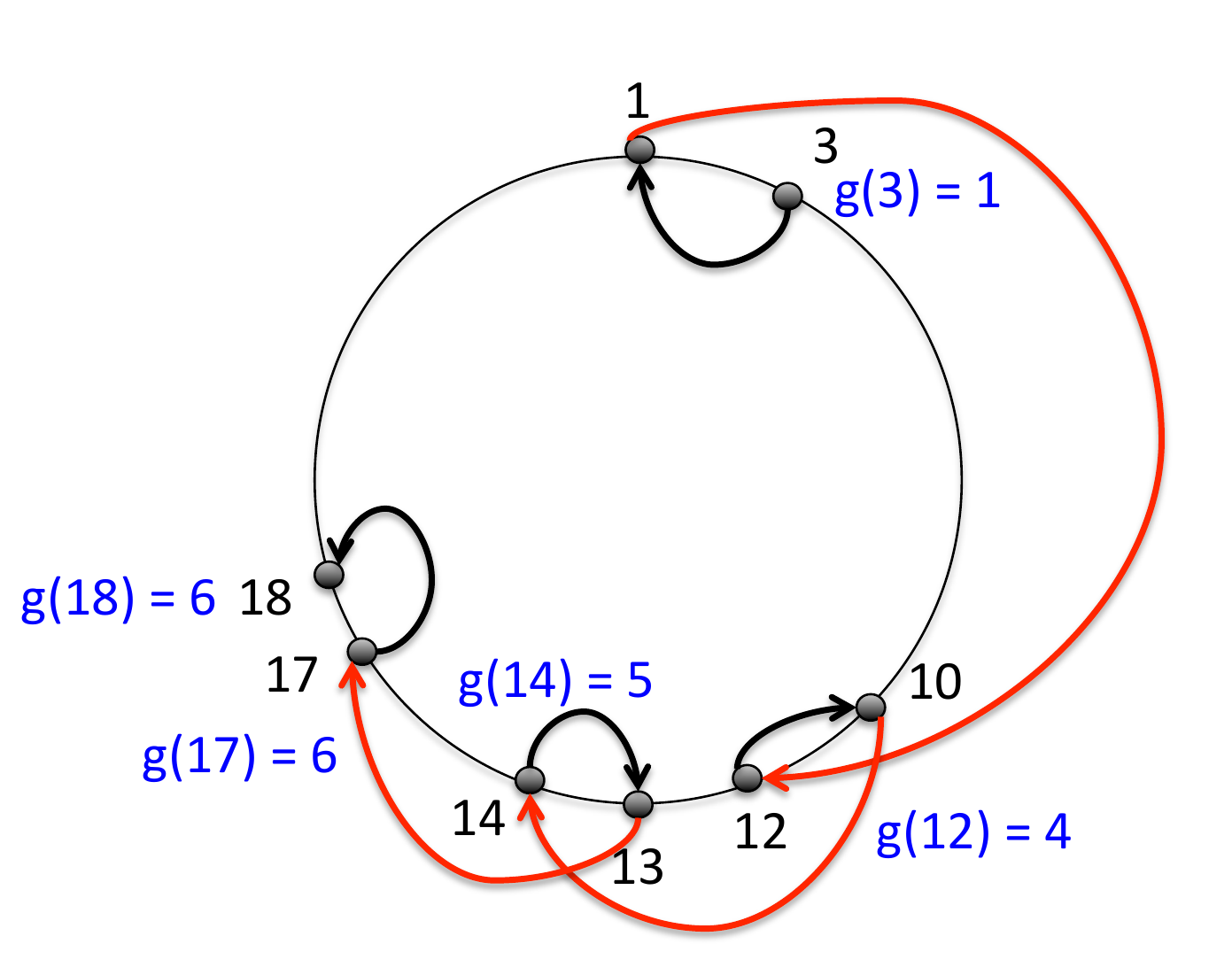}
\caption{Illustration of the routing strategy under LOG-2. Black arrows show transmissions in Tx-1 groups. Red arrows 
show transmissions in Tx-2 groups.}
\label{fig:hintlog_routing}
\end{figure}

In Fig. \ref{fig:hintlog_routing}, this procedure is illustrated for the $24$ node network of 
 Fig. \ref{fig:hintlog_example_table}. Suppose node $3$ wants to send to node $18$. 
Using Fig.  \ref{fig:hintlog_example_table}, we have that $g(3) = 1$ and $g(18) = 6$.
Thus, node $3$ transmits the packet to node $1$ in its Tx-1 group for relaying to $18$. 
Node $1$ checks if $18$ is in its Tx-2 group. Since it is not, node $1$ determines the node
in its Tx-2 group that covers node $18$. This node is given by the third node, i.e., node $12$, 
since $\mathcal{U}_{31} = \{4, 5, 6, 7, 8\}$. 
Therefore, node $1$ transmits the packet to node $12$ in its Tx-2 group. 
Node $12$ now repeats this procedure. Fig. \ref{fig:hintlog_routing} shows the final outcome of this process.
The entire route is given by $3$-$1$-$12$-$10$-$14$-$13$-$17$-$18$. Note that the hops alternate between Tx-1 and Tx-2 group
transmissions. Further, only the first node of a Tx-2 group transmits in any Tx-2 group transmission.

The following Lemma characterizes routing under LOG-2.
\begin{lem}
The LOG-2 routing strategy ensures that there are at most $2\log_2 N + 1$ hops between any pair of nodes.
\label{lem:hintlog_routing}
\end{lem}
\begin{proof} 
The proof is based on the observation that at each step in a Tx-2 
transmission, the distance between the node holding the packet and the destination
 decreases by \emph{at least} half. The full proof is omitted for brevity.
\end{proof}

\subsubsection{Throughput Analysis of LOG-2}
\label{section:hint_log_analysis}

We now show that the LOG-2 assignment along with the scheduling and routing strategy as described above
can achieve a throughput of ${1}/(\log_2M)^2$ for every node. 

\begin{thm} 
For $N = M\log_2M$ where $\log_2 M \in \mathbb{Z}^+$,  
LOG-2 can achieve a uniform per node throughput of ${1}/(\log_2M)^2$ using $2M$ channels.
\label{thm:hintlog_thruput}
\end{thm}

\begin{proof}
The proof is based on the following observations. Because of the symmetry of the assignment, it is 
sufficient to focus on the total load on the nodes in the first Tx-1 group and the first Tx-2 group.
Then we calculate a bound on the total number of group-to-group traffic flows
that involve these groups. This is used to show that a per node input rate of ${1}/(\log_2M)^2$
is feasible. The full proof is  provided in Appendix D.
\end{proof}

\subsection{Discussion of LOG-2}
\label{section:hint_log_discussion}

Theorem \ref{thm:hintlog_thruput} shows that LOG-2 can achieve a per node throughput of at least
${1}/{(\log_2(N))^2}$ with only $2$ transceivers per node.
Thus, we can get close to $\Theta(1)$ per node throughput with $O(1)$ transceivers per node.
However, this answers the first question raised in Sec. \ref{section:hintT_discussion}
only partially. As we show in the next section, under the Parking Lot model of Sec. \ref{section:network_model},
it is necessary to have $\Omega(N)$ channels to get $\Theta(1)$ 
per node throughput irrespective of the number of transceivers per node. 
However, given $\Omega(N)$ channels, it is not clear if one can achieve 
$\Theta(1)$ per node throughput using $O(1)$ transceivers per node under static channel assignment.
It is our conjecture that this is not possible.

\section{Efficiency of a Policy}
\label{section:tradeoff}

Let $\mathcal{P}$ denote the set of all possible feasible policies for channel assignment, 
routing, and scheduling under the network model described in Sec. \ref{section:network_model}. 
We show that there exists a simple relation between the total number of channels used, network size, 
and maximum per node throughput achievable under any policy $p \in \mathcal{P}$.
 Let $C_p$ denote the number of channels used by the channel assignment under $p$. 
 Let $\bar{L}_p$ denote the resulting average path length traversed by a packet under $p$. 
This average is taken over all source-destination pairs. 
 Finally, let $\lambda_p$ be the maximum per node input rate that can be supported under $p$.
Then, since the total traffic load cannot exceed the total transmission capacity of the network, we must have that
\begin{align}
N \lambda_p \bar{L}_p \leq C_p.
\label{eq:tradeoff}
\end{align}
The left hand side of (\ref{eq:tradeoff}) denotes the time average total number of transmissions 
required to deliver packets from all sources to all destinations. 
This cannot exceed $C_p$, the maximum number of transmissions possible per unit time.
 Since  $\bar{L}_p \geq 1$ under any $p$, it follows from  (\ref{eq:tradeoff}) that $C_p = \Omega(N)$ when
 $\lambda = \Theta(1)$. 
  
Note that (\ref{eq:tradeoff})  is a \emph{necessary} condition for the feasibility of any input rate $\lambda$ 
and can be used to establish a lower bound on the performance of any policy as shown next. 
\begin{defn}
Consider  any channel assignment, routing, and scheduling policy $p \in \mathcal{P}$ that uses $C_p$ channels in a Parking Lot network of size $N$. 
Suppose it can support a maximum per node input rate of $\lambda_p$. 
Then the efficiency of $p$, $\eta_p$, is defined as:
\begin{align}
\eta_p \defequiv \frac{N \lambda_p}{C_p}
\label{eq:gamma}
\end{align}
\label{defn:gamma}
\end{defn}
Since $\bar{L}_p \geq 1$ under any $p$, using (\ref{eq:tradeoff}) it follows that $\eta_p \leq 1$ for all $p \in \mathcal{P}$.
This includes the throughput maximizing policy obtained by solving the 
optimization problem for the tuple $(N, T, C_p)$ as discussed in Sec. \ref{section:network_model}.
Therefore, the maximum per node throughput achievable under a given policy $p$ that uses $C_p$ channels in a network of size $N$
is within a factor  $\eta_p$ of the maximum per node throughput achievable under \emph{any} policy  that uses $C_p$ channels
on the same network. Thus, the efficiency of a policy establishes a lower bound on its performance.

Using the results in Sec. \ref{section:hint} and \ref{section:towards}, the efficiency of the HINT-T and LOG-2 strategies can be calculated. 
Since HINT-T uses $TN^{\frac{T-1}{T}}$ channels to support per node rate of ${1}/{N^{\frac{1}{T}}}$, we have:
\begin{align}
\eta_{\textrm{HINT-T}} = \frac{N}{N^{1/T} \times TN^{\frac{T-1}{T}}} = \frac{1}{T}
\label{eq:gamma_hintT}
\end{align}
For a fixed $T$, we note that $\eta_{\textrm{HINT-T}}$ is \emph{independent} of the network size $N$. 
Thus, HINT-T is within a constant factor of the optimal solution.

Likewise, for $N = M \log_2 M$, LOG-2 uses $2M$ channels to support per node rate of $1/(\log_2M)^2$. This yields:
\begin{align}
\eta_{\textrm{LOG-2}} = \frac{N}{(\log_2M)^2 2M} = \frac{1}{2 \log_2M} > \frac{1}{2 \log_2N}
\label{eq:gamma_hintlog}
\end{align}
Thus, the LOG-2 scheme is within a logarithmic (in network size) factor of the optimal solution.
This may suggest that HINT-T has a better performance than LOG-2. However, note that for any given $T$,
the maximum per node throughput under HINT-T is $1/N^{1/T}$ while LOG-2 can achieve at least
${1}/{(\log_2 N)^2}$ with $T=2$ which exceeds  $1/N^{1/T}$ for sufficiently large $N$.

\section{Simulation-Based Evaluation}
\label{section:eval}

We now present simulation-based evaluation of HINT-T and LOG-2.
We also compare their performance against two representative schemes that we call
RING and GRID. These are described in the next section. 
In our simulations,  we consider a Parking Lot network with $4$ transceivers per node. 
Our simulations are performed using OPNET \cite{opnet}.

\subsection{RING and GRID Channel Assignment}
\label{section:ring_grid}

The RING assignment scheme first forms $N/4$ Tx-1 groups, each of size $4$ and containing consecutively numbered nodes. 
The other transceiver groups are obtained by shifting node ids in each Tx-1 group by $2, 3$ and $4$ respectively, 
with wraparound. Fig. \ref{fig:ring_example_table} shows this for $N=16$. The resulting topology resembles a ring 
(similar to Example $1$ in Fig. \ref{fig:example_assignment}), hence the name. 
The GRID scheme places nodes in a $2$-dim degree $4$ torus grid and assigns a channel to each link in the resulting graph.
Thus, each Tx group contains $2$ nodes.

Under the RING assignment, it is easy to show that the average path length over all source-destination pairs is $\Theta(N)$. Similarly, 
the total number of channels used is $N = \Theta(N)$.  Thus, using (\ref{eq:tradeoff}), it follows that the maximum
per node throughput under RING is $O(1/N)$. Further, using (\ref{eq:gamma}), we have that its efficiency is $\eta_{RING} = O(1/N)$.
Under the GRID assignment, the average path length over all source-destination pairs can be shown to be $\Theta(\sqrt{N})$ while 
the total number of channels used is $2N = \Theta(N)$.  Thus, using (\ref{eq:tradeoff}), it follows that the maximum
per node throughput under GRID is $O(1/\sqrt{N})$. Further, using (\ref{eq:gamma}), we have that 
its efficiency is $\eta_{GRID} = O(1/\sqrt{N})$.

\subsection{Simulation Setup}
\label{section:sim_setup}

\begin{figure}
\centering
\includegraphics[height=2.5cm, angle=0]{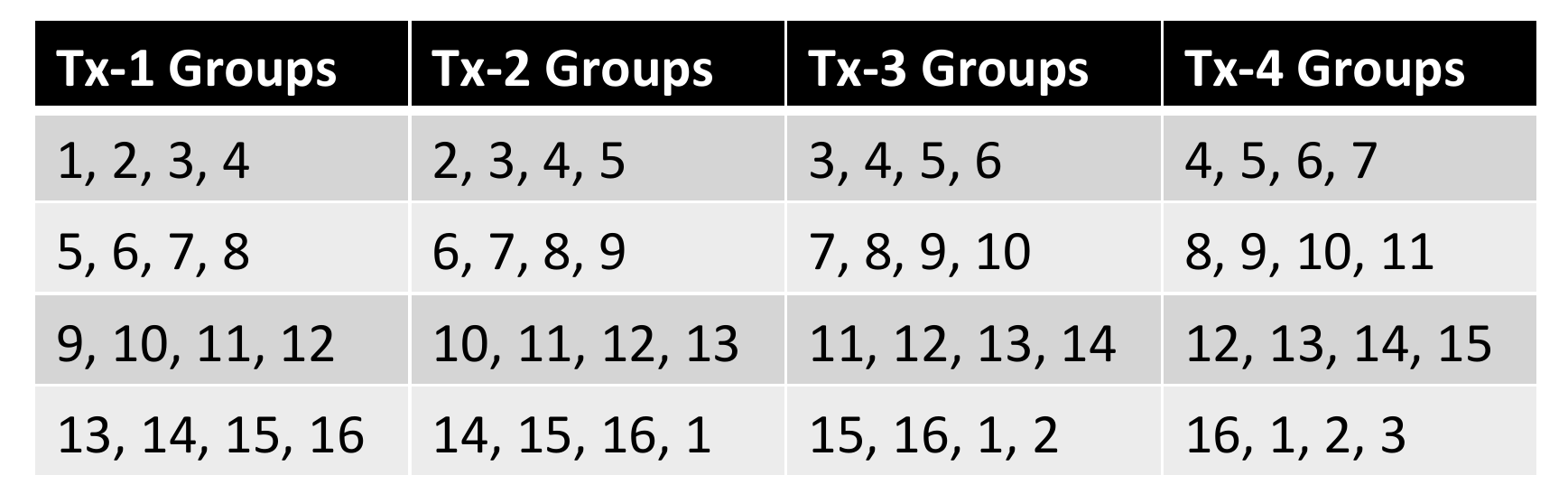}
\caption{Channel Assignment under RING for $N = 16$.}
\label{fig:ring_example_table}
\end{figure}

The input topology is a Parking Lot MC-MR network of size $N$ nodes where $N$ is varied between $16$ and $100$.
These nodes are placed in a $250 m \times 250 m$ area and use a fixed transmit power such that all nodes are within 
the transmission range of each other. Every node has $4$ identical transceivers, each capable of operating on a
channel of bandwidth $10$ MHz. We assume that such channels are available starting from $900$ MHz. 
Each transceiver independently uses 802.11 CSMA for medium access on its assigned frequency channel. 
The raw MAC level throughput per channel is $1.2$ Mbps.

Under each scheme, the assignments to the different transceivers are determined at the start of a simulation run and
fixed for the duration of that run. In each run, traffic is generated using the uniform all-pair unicast model.
Specifically, every node generates packets according to a Poisson process of fixed rate and each packet is equally likely to be
destined to any of the other nodes. We assume that all packets are fixed size UDP packets of length $436$ bytes that includes
control headers. Each simulation run has a duration of $150$ seconds at the end of which we count the total number of packets 
delivered successfully to their destinations. We use the default link state routing protocol available in OPNET for routing packets. 
We do not implement the load balanced routing strategies under HINT and LOG as discussed in the earlier sections.
Thus, we expect the achievable performance of HINT and LOG with load balanced routing to outperform what we report here.

Our objective is to compare the maximum  achievable per node throughput  under each scheme. In addition, we are interested in 
comparing it against the theoretical bounds. Given a network size $N$ and an assignment scheme $X$, we estimate 
it as follows. 
Simulations are run with increasing values of the input rate until when the total number of delivered packets does not increase anymore. 
We call this the ``saturation throughput'' under $X$ for a given $N$ and use it as a measure of the 
maximum per node achievable throughput under $X$ for $N$.
Note that the theoretical bounds for maximum throughput are derived under idealized assumptions such as 
collision free transmission scheduling, load balanced routing, no buffer overflows, and ignoring any control overheads (such as due to
CSMA and link state routing updates). This no longer holds in the simulations which are closer to the realistic setting.
However, these factors affect all of the schemes being compared and the  saturation throughput can be thought of as 
a measure of the remaining effective capacity. 


\begin{table}
\scriptsize
\centering
 \begin{tabular}{|c|c|c|c|} \hline
 Algorithm             &   $\lambda_{p}$                              &   $C_{p}$           &    $\eta_p = \frac{N \lambda_{p}}{C_{p}}$  \\ \hline
RING                   &   $O\big(\frac{1}{N}\big)$                     &   $O(N)$             &    $O\big(\frac{1}{N}\big)$   \\ \hline
GRID                   &   $O\big(\frac{1}{\sqrt{N}}\big)$            &   $O(N)$             &    $O\big(\frac{1}{\sqrt{N}}\big)$  \\ \hline
2$\times$HINT-2 &   $O\big(\frac{1}{\sqrt{N}}\big)$            &	  $O(\sqrt{N})$    &    $O(1)$ \\ \hline
2$\times$LOG-2  &   $O\Big(\frac{1}{(\log_2{N})^2}\Big)$   &   $O\big(\frac{N}{\log_2{N}}\big)$  &  $O\big(\frac{1}{\log_2{N}}\big)$  \\ \hline
HINT-4                 &  $O\Big(\frac{1}{N^{1/4}}\Big)$             &    $O(N^{3/4})$    &  $O(1)$    \\ \hline
 \end{tabular}
\caption{Summary of Theoretical performance bounds.}
\label{table:efficiency}
\end{table}

In this setup, we implement the following: RING, GRID, 
2$\times$HINT-2, 2$\times$LOG-2, and HINT-4. 
2$\times$HINT-2 implements HINT-2 on transceivers $1, 2$ and repeats it on transceivers $3, 4$.
Similarly, 2$\times$LOG-2 implements LOG-2 on transceivers $1, 2$ and repeats it on transceivers $3, 4$.
Table \ref{table:efficiency} summarizes the theoretical performance  bounds for these schemes.

\subsection{Simulation Results and Discussion}
\label{section:sim_results}

We first compare these schemes in terms of their saturation throughput.
Fig. \ref{fig:sim_results}(a) plots the per node saturation throughput  (in packets/sec, pps) vs. N. 
As can be seen, HINT-4 outperforms the RING scheme by $200$-$300\%$
for high tens of nodes, and outperforms its nearest rival GRID by nearly $150\%$ at $N=100$.
The behavior of the curves is consistent with the theoretical bounds with RING showing the largest drop as $N$ increases.  
Note that the per node throughput  under 2$\times$LOG-2 scales as $\Theta(1/(\log_2 N)^2)$ which is the best scaling performance 
among all of these schemes. However, in the finite range of $N$ considered here, HINT-4 outperforms 2$\times$LOG-2. Both GRID and 
2$\times$HINT-2 also outperform 2$\times$LOG-2 up to a crossover point, after which 2$\times$LOG-2 is better. We expect
2$\times$LOG-2 to eventually outperform HINT-4 as well after a sufficiently large $N$.
 
We note that GRID generally has a better performance than both 2$\times$HINT-2 and 2$\times$LOG-2. 
However, this plot only considers the
per node saturation throughput and does not capture the total number of channels used.  To incorporate this, 
we next compare these schemes in terms of their efficiency.
Fig. \ref{fig:sim_results}(b) plots the total network saturation throughput vs. the number of channels used under each scheme. 
Recall that efficiency of a scheme is the ratio between maximum total network 
throughput and the total number of channels used. Thus, the slope of the performance curve for a scheme in 
Fig. \ref{fig:sim_results}(b) can be interpreted as its efficiency at that point.
It can be seen that both 2$\times$HINT-2 and 2$\times$LOG-2 have much higher efficiency than GRID. Further, this gap
increases with $N$.

Fig. \ref{fig:sim_results}(b)  also agrees quite well with the theoretical bounds. 
2$\times$HINT-2 and HINT-4 both have a theoretical efficiency of $O(1)$, i.e., independent of $N$ (see Table \ref{table:efficiency}). 
Their curves in Fig. \ref{fig:sim_results}(b) have a slope that does not decrease with $N$. 
On the other hand, RING has an efficiency of $O(1/N)$ which goes to $0$ as $N$ increases. 
Its curve in Fig. \ref{fig:sim_results}(b) flattens to slope $0$ quickly. 
GRID has efficiency $O(1/\sqrt{N})$. Its curve also flattens, but more slowly than RING. 
Finally, 2$\times$LOG-2 has efficiency $O(1/\log_2 N)$ and its slope is even  more slower
to flatten.

\begin{figure}[t]
\centerline{\hbox{
     \subfigure[]{\includegraphics[height=3cm, width=1.87in]{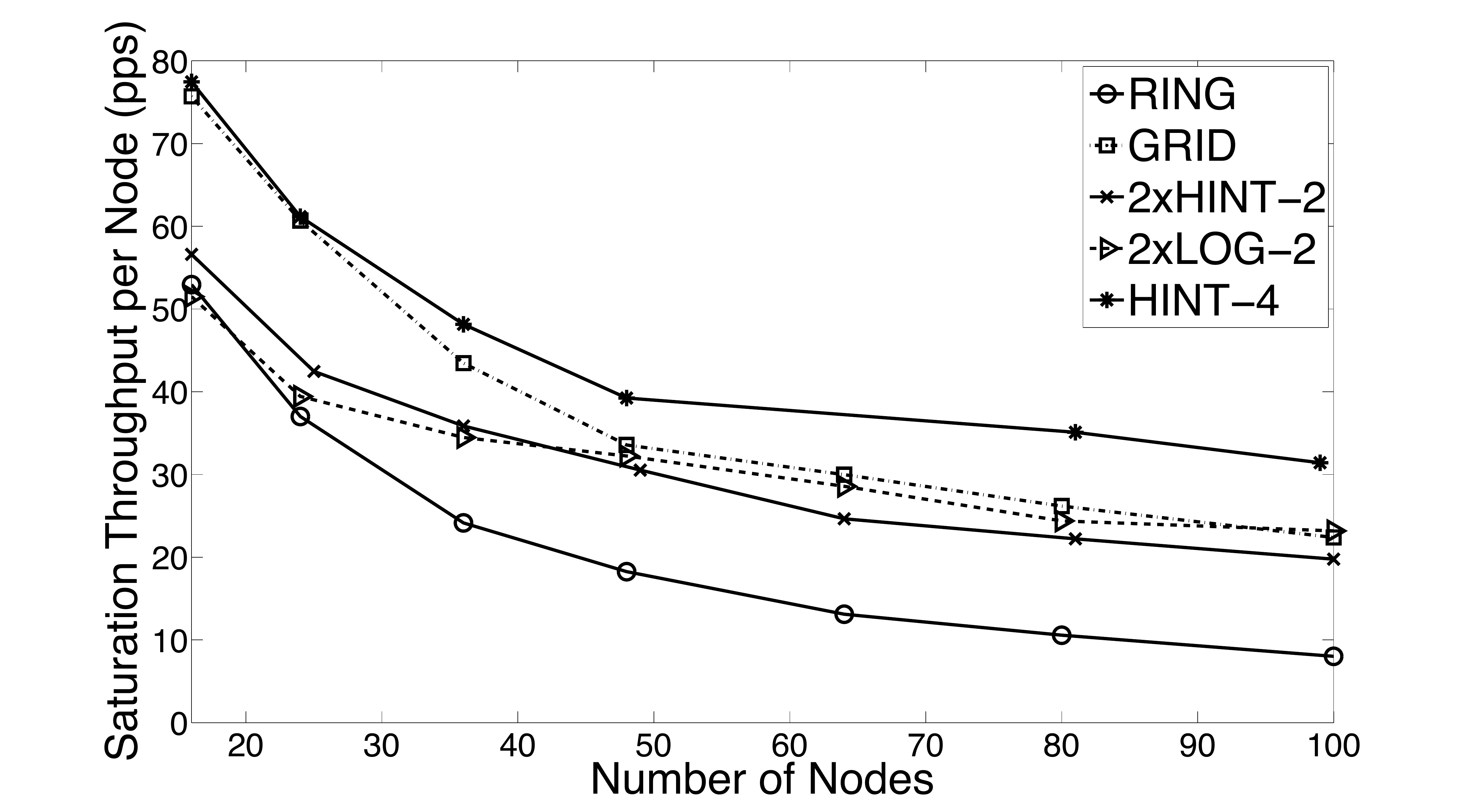}}
     \subfigure[]{\includegraphics[height=3cm, width=1.87in]{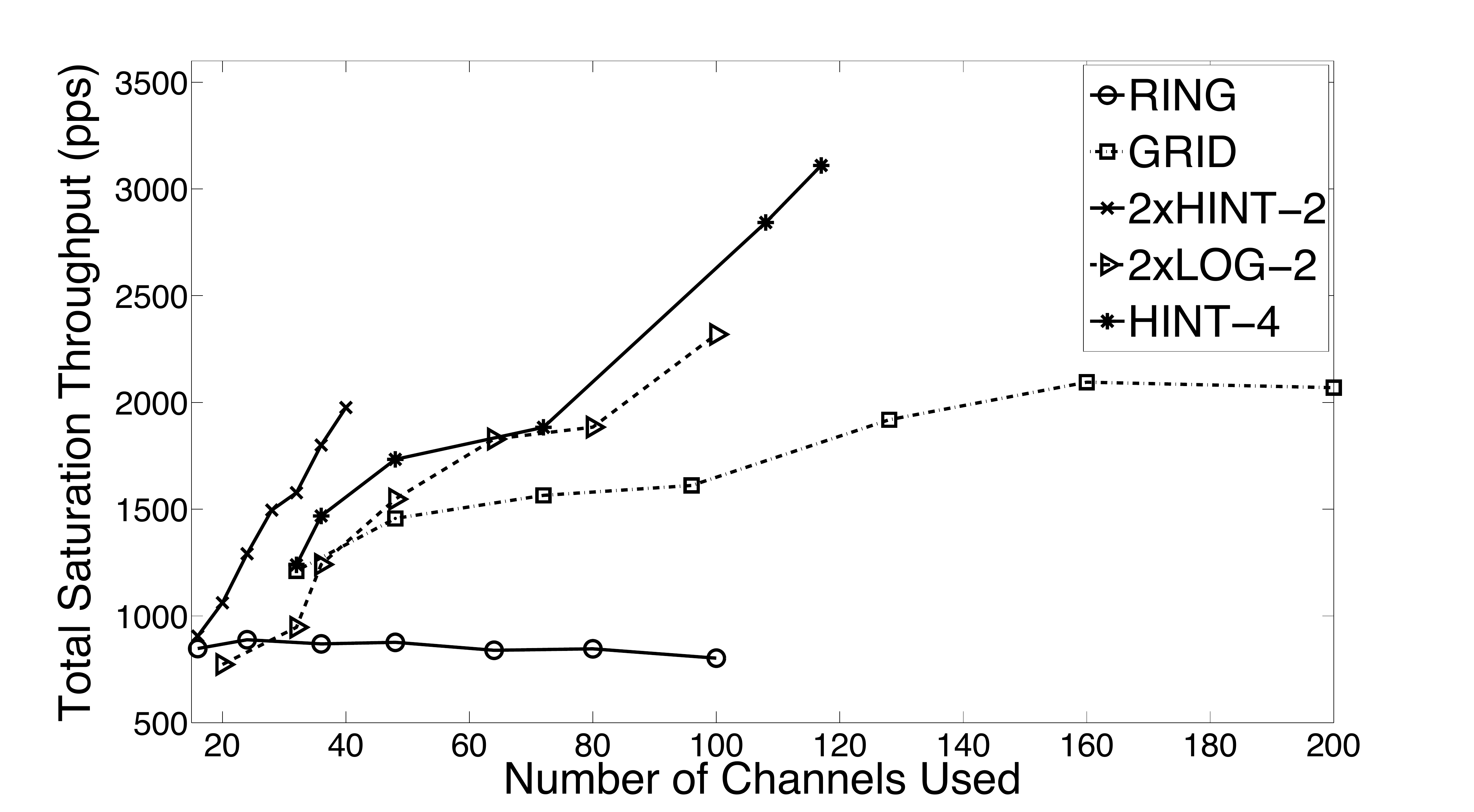}}}}
\caption{(a) Saturation Throughput per node vs. Network Size under different schemes. (b) Total Saturation Throughput 
vs. Number of Channels used under different schemes. The slope of each curve is the efficiency of that scheme at that point.}
\label{fig:sim_results}
\end{figure}

In terms of both per node throughput and efficiency, RING has the worst performance. 
Intuitively, this is because under RING, packets have to traverse $O(N)$ hops on average, 
resulting in a vast majority of traffic being relay traffic.
GRID improves upon RING because the average distance between nodes is now $O(\sqrt{N})$. It 
has similar scaling as HINT-2 in terms of throughput but uses a lot more channels. Thus, it has poor efficiency.
The HINT-T schemes have the best performance in terms of their efficiency which does not depend on $N$.
However, in terms of throughput, $T$ must be large to remain close to $O(1)$ as $N$ increases.
LOG-2 can achieve this with just $2$ transceivers. However, its efficiency is not as good as HINT-T.
 Thus, both HINT and LOG are order optimal or close to order optimal along one of the dimensions (throughput or efficiency).

\section{Conclusions and Future Work}
\label{section:conc}

\subsection{Implications for Network Architecture}
\label{section:implications}

Our results have several  interesting implications on the architecture of dense MC-MR networks.
The scaling laws in \cite{GK} imply that in a random ad hoc network, if individual transceivers can only operate over channels 
of fixed bandwidth (that does not increase with $N$), 
then the best possible scaling is $O(1/\sqrt{N})$ even if the network has a large number of channels available
 and each node has multiple (but finite) transceivers. 
The idea behind the optimal strategy in \cite{GK} is to keep transmit power sufficiently low, just enough to ensure connectivity,
 and use multi-hop routing. The motivation behind reducing the transmit power is to maximize spatial reuse of the finite network bandwidth.
 
In contrast, our results show that in the Parking Lot model, it is possible to get significantly 
higher throughput if nodes have multiple transceivers and the number of available  channels 
scales with the network size.
Our results also suggest that in this model, reducing power to maximize spatial reuse of frequency does not help. It is better to preserve the Parking Lot
structure and use all available channels while keeping inter-node distance small in the resulting effective topology by careful channel assignment.

\subsection{Open Questions}
\label{section:open}

There are several open questions that offer directions for future work.
For example, it is not clear if one can achieve 
$\Theta(1)$ per node throughput using $O(1)$ transceivers per node under static channel assignment. Similarly, it is not
clear which algorithm is the best given a particular $(N, T, F)$. 
Recall that HINT-T has $\Theta(1)$ efficiency for any fixed $T$ but its throughput scales as $\Theta(1/N^{1/T})$.
On the other hand, LOG-2 has $\Theta(1/(\log_2 N)^2)$ throughput but $\Theta(1/(\log_2 N))$ efficiency. Thus,
it would be interesting to design a ``universal'' channel assignment policy that has both high throughput and high efficiency for finite $T$ 
at \emph{all} points of the curve (\ref{eq:tradeoff}). 


\section*{Appendix A \\ Proof of Theorem \ref{thm:hint2_thruput}}
\label{section:appxA}

Because of the symmetry of the assignment, it is sufficient to focus on the total load on each of the two transceivers of node $1$ and
show that it can be supported. First, consider the second transceiver of node $1$. On this, 
node $1$ transmits those packets generated by itself that are  destined for nodes in all Tx-1 groups other than its own group. 
There are $M-1$ such groups and each group has $M$ nodes. 
Node $1$ generates packets at rate $\frac{\lambda}{N-1}$ for each of these nodes.
Thus, the total traffic load on the second transceiver of node $1$ is given by:
\begin{align}
\frac{\lambda M (M -1)}{N-1}
\label{eq:hint2_1}
\end{align}
This must be less than the total transmission rate $1/M$. 
Solving this, we get:
\begin{align}
\lambda \leq \frac{N-1}{M^2(M - 1)} = \frac{M^2-1}{M^2(M - 1)} = \frac{1}{M} + \frac{1}{M^2}
\label{eq:hint2_2}
\end{align}
From this, it can be seen that $\lambda = 1/M$ satifies (\ref{eq:hint2_2}).

Next consider the first transceiver of node $1$. 
On this, node $1$ transmits packets generated by itself that are destined for all the other nodes in its own Tx-1 group. 
There are $M - 1$ such nodes. In addition, node $1$ transmits the relay packets that it received for these nodes from its neighbors in its Tx-2 group.
There are $M-1$ such neighbors.
Thus, the total traffic load on the first transceiver of node $1$ is given by:
\begin{align}
\frac{\lambda (M -1)}{N-1} + \frac{\lambda (M -1) (M -1)}{N-1}  = \frac{\lambda M(M-1)}{N - 1}
\label{eq:hint2_3}
\end{align}
This is same as (\ref{eq:hint2_1}). 
Again, this must be less than the total transmission rate $1/M$.
Using the same steps as before, we have that
$\lambda = 1/M$ is feasible for this case as well. 
Since HINT-2 uses $2M$ channels and $M=N^{1/2}$, Theorem \ref{thm:hint2_thruput} follows.


\section*{Appendix B \\ Proof of Lemma \ref{lem:hintT_routing}}
\label{section:appxB}



To show Lemma \ref{lem:hintT_routing}, it is sufficient to show that 
$k(r^*, d) < k(s, d)$. This is because this means $k(r, d)$ will
eventually become $1$ and that means $r$ is in the same group as $d$. Further, this will take
no more than $T$ hops since $k(s, d) \leq T$.

Consider the Tx-k(s, d) group of node $s$. By definition, $k(s,d)$ is the smallest $k$ for which 
there exists a set $\mathcal{S}_{ik}$ such that both $s$ and $d$ are in $\mathcal{S}_{ik}$. Denote
this set by $\mathcal{S}_{j,k(s,d)}$.
Note that this set is a union of $M$ disjoint sets from level $k(s,d) - 1$. Further,
the Tx-k(s, d) group of node $s$ contains $M$ nodes (including $s$), one from each of 
these $M$ sets from level $(k(s,d) - 1)$.  This follows by the HINT-T construction.

Now the level $(k(s,d) - 1)$ set that contains $s$ cannot have $d$. Otherwise, $k(s, d)$ 
cannot be the smallest $k$ for which
there exists a set $\mathcal{S}_{ik}$ such that both $s$ and $d$ are in $\mathcal{S}_{ik}$.

However, there must be one set of the remaining $M-1$ level $(k(s,d) - 1)$ sets that contains
$d$. Othwerise, $d$ cannot be in $\mathcal{S}_{j,k(s,d)}$, a contradiction.
Now there is one node from this set in the Tx-k(s, d) group of node $s$. Call it $n$. Since
$d$ is in the same level $(k(s,d) - 1)$ set as $n$, we must have that $k(n, d) \leq k(s,d) - 1$.
Thus, $k(r^*, d) < k(s, d)$.

\section*{Appendix C \\ Throughput Analysis of HINT-T}}
\label{section:appxC}


Because of the symmetry of the CA and routing strategy across nodes, it is sufficient to focus on 
the total load on each of the $T$ transceivers of node $1$ and show that it can be supported. 

Consider the $k^{th}$ transceiver of node $1$. At this level, node $1$ is responsible for forwarding packets destined to
$(M-1)M^{k-1}$ nodes using its $(M-1)$ neighbors in its Tx-k group. 
Next, we calculate the total number of source nodes that generate packets that are destined to these nodes and that are routed
via the $k^{th}$ transceiver of node $1$. We calculate this number for each level, starting at level $k$ and going upto level $T$,
 making sure that no node is counted more than once.

At level $k$, $1$ is the only such source node. At level $k+1$ there are $M-1$ new source nodes, each a neighbor of
node $1$ in its Tx-k+1 group. At each next level, all the nodes counted so far bring $M-1$ new (previously not counted) nodes. 
Thus, at level $k+2$, we have $(1 + M - 1)(M-1) = M(M-1)$ new source nodes. Similarly, at level $k+3$, we have
$(1 + M - 1 + M(M-1))(M-1) = M^2(M-1)$ new source nodes. In general, it can be shown that at level $k+i$,
we have $M^{i-1}(M-1)$ new soure noes. Summing from level $k$ through $T$, the total number of  source nodes is given by:
\begin{align}
& 1 + \sum_{i=1}^{T-k} M^{i-1}(M-1) = M^{T-k}
\label{eq:apdxA_1}
\end{align}

Thus, the total traffic load on the $k^{th}$ transceiver of node $1$ is given by:
\begin{align}
\frac{\lambda M^{T-k} (M-1)M^{k-1}}{M^T -1} = \frac{\lambda (M-1)M^{T-1}}{M^T -1}
\label{eq:apdxA_2}
\end{align}
This cannot exceed the available  transmission rate ${1}/{M}$. 
Solving this, we get
\begin{align}
\lambda \leq \frac{M^T - 1}{M^T (M-1)}
\label{eq:apdxA_3}
\end{align}
It can be seen that $\lambda = 1/M$ satisfies this.


\section*{Appendix D \\  Throughput Analysis of LOG-2}
\label{section:appxD}


Because of the symmetry of the assignment, it is sufficient to focus on the total load on the nodes in the first
Tx-1 group and the first Tx-2 group.

We calculate the total number of group-to-group flows that involve the 
first Tx-1 group (as a source, destination, or relay).

Consider all the traffic flows that originate in the first Tx-1 group and are destined to all nodes. 
Each flow involves a sequence of transmissions where each transmission takes place either in a Tx-1 group
or a Tx-2 group. Note that no group occurs more than once in any routing path.

Let $\Delta_i$ denote the sum of the path lengths (measured in the units hops or transmissions involving only  Tx-1 groups)
of flows originating in the $i^{th}$ Tx-1 group and destined to nodes in all other groups.
Since the paths for flows originating in any other Tx-1 group can be obtained by simply shifting the 
paths for flows originating in the first Tx-1 group, we have that $\Delta_i = \Delta_1$ for all $i$.
Further, the sum of the number of times the first Tx-1 group occurs across all the paths is also equal to $\Delta_1$.
Finally, this equals the total number of 
group-to-group flows that involve the 
first Tx-1 group (as a source, destination, or relay).

Next, note that the maximum path length (measured in the units hops or transmissions involving only  Tx-1 groups)
is $(\log_2 M + 1)$. This implies that $\Delta_1 \leq M (\log_2 M + 1)$. Thus, the total load on the first Tx-1 group  is upper bounded by

\begin{align}
M (\log_2 M + 1) \lambda_g = M (\log_2 M + 1) \frac{(\log_2 M)^2 \lambda}{M\log_2M - 1}
\label{eq:apdxB_1}
\end{align}
Since this cannot exceed the total transmission capacity of the  first Tx-1 group,we must have that

\begin{align}
 M (\log_2 M + 1) \frac{(\log_2 M)^2 \lambda}{M\log_2M - 1} \leq 1
\label{eq:apdxB_2}
\end{align}

It can be seen that $\lambda = 1/(\log_2 M)^2$ satisfies this.

\end{document}